\documentclass[preprint,number,sort&compress]{elsarticle}
\usepackage{amsmath, amsthm, amssymb}
\usepackage{amsbsy}

\usepackage{graphicx}
\usepackage{color}

\renewcommand{\vec}[1]{\mathbf{#1}}
\newcommand{\gvec}[1]{\boldsymbol{#1}}

\newcommand{\ellts}{_{L^2}^2}
\newcommand{\ellt}{_{L^2}}
\newcommand{\rr}{{\mathbb R}}
\newcommand{\uu}{\vec{u}}
\newcommand{\ff}{\vec{f}}
\newcommand{\pp}{{\mathrm P}}

\newtheorem{theorem}{Theorem}
\newtheorem{lemma}{Lemma}

\journal{Physica D}

\begin{document}

\begin{frontmatter}

\title{Turbulence properties and global regularity of a modified Navier-Stokes equation}

\author[rub]{Tobias Grafke}
\author[rub]{Rainer Grauer}
\ead{grauer@tp1.rub.de}
\author[ucsb]{Thomas C. Sideris}

\address[rub]{Institut f\"ur Theoretische Physik I, Ruhr-Universit\"at Bochum, Germany}
\address[ucsb]{Department of Mathematics, University of California, Santa Barbara, USA}

\begin{abstract}
  We introduce a modification of the Navier-Stokes equation that has the
  remarkable property of possessing an infinite number of conserved quantities
  in the inviscid limit. This new equation is studied numerically and turbulence
  properties are analyzed concerning energy spectra and scaling of structure
  functions. The dissipative structures arising in this new equation are curled
  vortex sheets contrary to vortex tubes arising in Navier-Stokes turbulence.
  The numerically calculated scaling of structure functions is compared with a
  phenomenological model based on the She-L\'ev\^eque approach.

  Finally, for this equation we demonstrate global well-posedness for
  sufficiently smooth initial conditions in the periodic case and in $\mathbb
  R^3$. The key feature is the availability of an additional estimate which
  shows that the $L^4$-norm of the velocity field remains finite.
\end{abstract}


\end{frontmatter}


\section{Introduction}

In this paper, we introduce a new equation which is a hybrid of the
Navier-Stokes equation and the Burgers equation. This equation possesses a
nonlocal part in the nonlinearity analogous to the pressure term in the
Navier-Stokes equation. This new equation has the remarkable feature of
possessing an infinite number of conserved quantities in the inviscid limit.

Turbulence properties of this equation are analyzed using numerical simulations.
We calculate energy spectra and scaling of higher order structure functions. A
key observation from the numerical simulations of this modified Navier-Stokes
equation is that the most dissipative structures consist of curled vortex sheets
contrary to vortex tubes in conventional Navier-Stokes turbulence. Using this
information, a She-L\'ev\^eque type model is derived and compared to the numerically
obtained scaling of higher order structure functions.

In addition, for this new equation we can show existence and regularity for
$H^1$ initial conditions of arbitrary size.  This will be carried out in
${\mathbb R}^3$ and periodic domains in ${\mathbb R}^3$. The simple modification
of the nonlinearity makes the proof of global solutions possible, insofar as an
additional estimate is available showing that the solution remains finite in
$L^p$, $2<p<\infty$. With $p=4$, this is then coupled with standard estimates
for the $H^1$-norm to complete the proof.

The outline of this paper is as follows: in section 2 we motivate and introduce
our new equation. Turbulence statistics and phenomenological modeling is
considered in section 3. Section 4 contains the proof of existence of global
solutions. We finish with remarks on possible further consequences of the
existence of an infinite number of conserved quantities.


\section{Model equation}

We consider  a three-dimensional domain $\Omega$
which shall be either $\mathbb R^3$ or a bounded cube in $\mathbb R^3$ with
periodic boundary conditions.
Let
$\mathrm{P} = 1 -  \Delta^{-1} \nabla \otimes\nabla$ be the Leray-Hopf
projection operator (with periodic boundary conditions when $\Omega$ is bounded):
\begin{equation}
	\mathrm{P} [\mathrm{P} [ \vec{u} ]] = \mathrm{P} [ \vec{u} ] \;\; , \;\;\;
	\nabla\cdot \mathrm{P} [ \vec{u} ] = 0 \; .
\end{equation}
The usual incompressible Navier-Stokes equation
\begin{displaymath}
	\frac{\partial}{\partial t} \vec{v} + \vec{v}\cdot\nabla\vec{v} + \nabla p = \nu \Delta \vec{v}+\vec{f}
	\;\; , \;\;\; \nabla\cdot \vec{v} = 0
\end{displaymath}
can be written with the projection operator $\mathrm{P}$
\begin{displaymath}
	\frac{\partial}{\partial t} \vec{v} + \mathrm{P} [\vec{v}\cdot\nabla\vec{v}] = \nu \Delta
	\vec{v} +\pp[\vec{f}]\;\; , \;\;\; \nabla\cdot \vec{v} = 0
\end{displaymath}
such that no explicit pressure term is present in the equation.

We can rewrite the Navier-Stokes equation without the
incompressibility constraint in the form
\begin{equation}
  \label{eq:ns}
  \frac{\partial}{\partial t} \vec{u} +  \mathrm{P} [ \vec{u} ] \cdot \nabla  \mathrm{P} [ \vec{u} ] 
  = \nu \Delta \vec{u} + \vec{f}\;,
\end{equation}
where the solution of the Navier-Stokes equation can be recovered by taking $\vec{v}=P[\vec{u}]$.

The  equation (\ref{eq:ns})  can be compared with Burgers equation
whose structure is formally similar:
\begin{equation}
  \label{eq:bu}
  \frac{\partial}{\partial t} \vec{u} +  \vec{u} \cdot \nabla  \vec{u} = \nu \Delta
  \vec{u} + \vec{f} \; .
\end{equation}
For equation (\ref{eq:bu}) the nonlinearity is purely local, whereas
for equation (\ref{eq:ns}) the nonlinear interaction involves the
nonlocal projection.

A natural hybrid of these two equations leads a new model equation
involving a compressible velocity field $\vec{u}$ that is convected by
its solenoidal part $\mathrm{P} [ \vec{u} ]$:
\begin{equation}
  \label{eq:eb}
  \frac{\partial}{\partial t} \vec{u} + \mathrm{P} [ \vec{u} ] \cdot \nabla \vec{u} = \nu \Delta
  \vec{u} + \vec{f} \;.
\end{equation}
More accurately this means: The convection of the velocity field
$\vec{u}$ is local in position space, but the projection operator is
local in Fourier space and thus shares this mixture of local and
non-local interactions with the original Navier-Stokes equation.


\section{Turbulence statistics}

By construction the presented model equation is an intermediate step
between the Navier-Stokes and Burgers equation, which in turn differ
significantly in their dynamical evolution and turbulent behavior. In
Navier-Stokes turbulence, on the one hand, the most dissipative
structures are vortex filaments, while for Burgers equation shocks
dominate the turbulent flow. It is of obvious interest in how far our
model equation bridges between those, which structures are the most
dominant for turbulent flows and how these structures influence the
turbulence statistics. We therefore extend the She-L\'ev\^eque
reasoning, which describes Navier-Stokes and Burgers turbulence well,
to our model equation and test it against numerical simulations by
comparing the scaling exponents of the structure functions.

Numerical simulations are carried out with a second-order in space
finite difference scheme with a strongly stable third-order
Runge-Kutta time integration with resolutions up to $512^3$. The
initial conditions were chosen as Orszag-Tang-like (see
\cite{politano-pouquet-sulem:1995}) large-scale perturbations:
\begin{eqnarray}
\label{eq:Nsc}
  u_x &=& A \left( -2 \sin(2y) + \sin(z) + 2 \cos(2y) + \cos(z) \right) \nonumber \\
  u_y &=& A \left( -2 \sin(x) + \sin(z) + 2 \cos(x) + \cos(z) \right) \nonumber \\
  u_z &=& A \left( \sin(x) + \sin(y) -2 \cos(2x) + \cos(y) \right). \nonumber
\end{eqnarray}
For simplicity and comparability both velocity and its solenoidal
projection are set to equal values. The physical domain stretches from
$-\pi$ to $\pi$; the above defined conditions, thus, are both
large-scale perturbations and periodic. All hydrodynamical models will
be simulated in comparison, using these initial conditions. We
consider only decaying turbulence without external forces. For the
parameters of all performed runs see Table~\ref{tab:runs}, which shows
the numerical value of the quantities at the time of maximum enstrophy
$t=t_\mathcal{E}$.

\begin{table}[h]
  \centering
  \begin{tabular}{c|ccccccccc}

    PDE     & $N$ & $\Delta x$ & $t_\mathcal{E}$ & $\nu$  & $v_\mathrm{rms}$ & $\varepsilon$ & $L$     & $\eta$   & $R_\lambda$ \\
    \hline \hline
    Burgers & 128 & 0.0491     & 0.461           & 0.0241 & 1.761            & 3.564         & 1.533   & 0.0445   & 40.99       \\
            & 256 & 0.0245     & 0.603           & 0.015  & 1.637            & 3.320         & 1.320   & 0.0318   & 46.49       \\
            & 512 & 0.0123     & 0.634           & 0.009  & 1.656            & 3.361         & 1.352   & 0.0216   & 61.10       \\
    \hline                                                         
    Model   & 128 & 0.0491     & 1.438           & 0.01   & 1.316            & 0.983         & 2.318   & 0.0318   & 67.65       \\
    equation& 256 & 0.0245     & 1.459           & 0.006  & 1.369            & 1.044         & 2.459   & 0.0213   & 91.75       \\
            & 512 & 0.0123     & 1.624           & 0.0036 & 1.370            & 1.096         & 2.344   & 0.0144   & 115.7       \\
    \hline                                                         
    Navier- & 128 & 0.0491     & 2.502           & 0.007  & 1.482            & 0.918         & 3.549   & 0.0247   & 106.2       \\
    Stokes  & 256 & 0.0245     & 2.616           & 0.00278& 1.518            & 1.268         & 2.761   & 0.0114   & 150.4       \\
            & 512 & 0.0123     & 2.545           & 0.00110& 1.551            & 1.572         & 2.376   & 0.00539  & 224.2       \\
    \hline \hline
  \end{tabular}
  \caption[Overview for all simulations]{
    Parameters of the numerical simulations. 
    number of collocation points $N^3$;
    grid spacing $dx$;
    time of fully developed turbulence $t_\mathcal{E}$;
    viscosity $\nu$;
    root-mean-square velocity $v_{rms} = \sqrt{2/3 E_{kin}}$;
    mean energy dissipation rate $\varepsilon$;
    integral scale $L = (2/3 E_{kin})^{3/2}/\varepsilon$;
    dissipation length scale $\eta = (\nu^3/\varepsilon)^{1/4}$;
    Taylor-Reynolds number $R_\lambda = \sqrt{15 v_{rms} L / \nu}$;
    all taken at the time of maximum enstrophy $t=t_\mathcal{E}$.
    }
  \label{tab:runs}
\end{table}

\begin{figure}[h]
\begin{center}
  \includegraphics[width=0.45\textwidth]{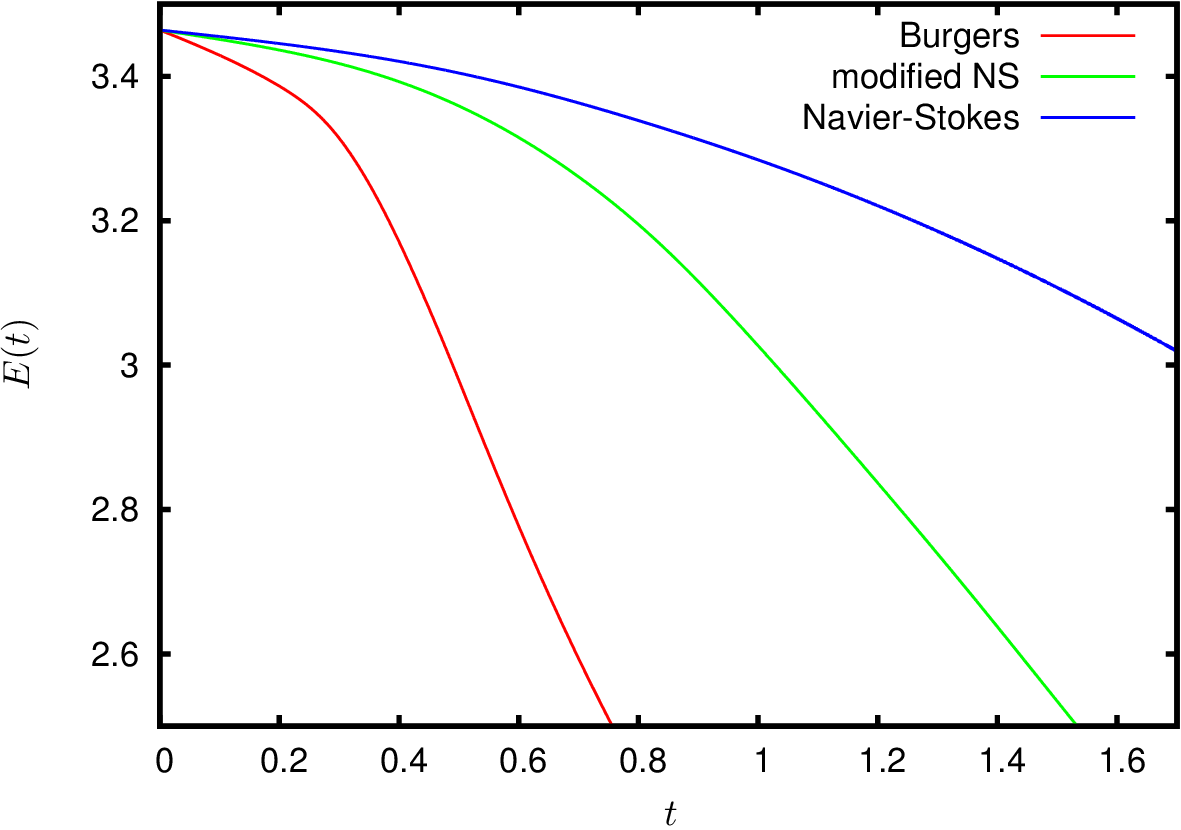}
  \includegraphics[width=0.45\textwidth]{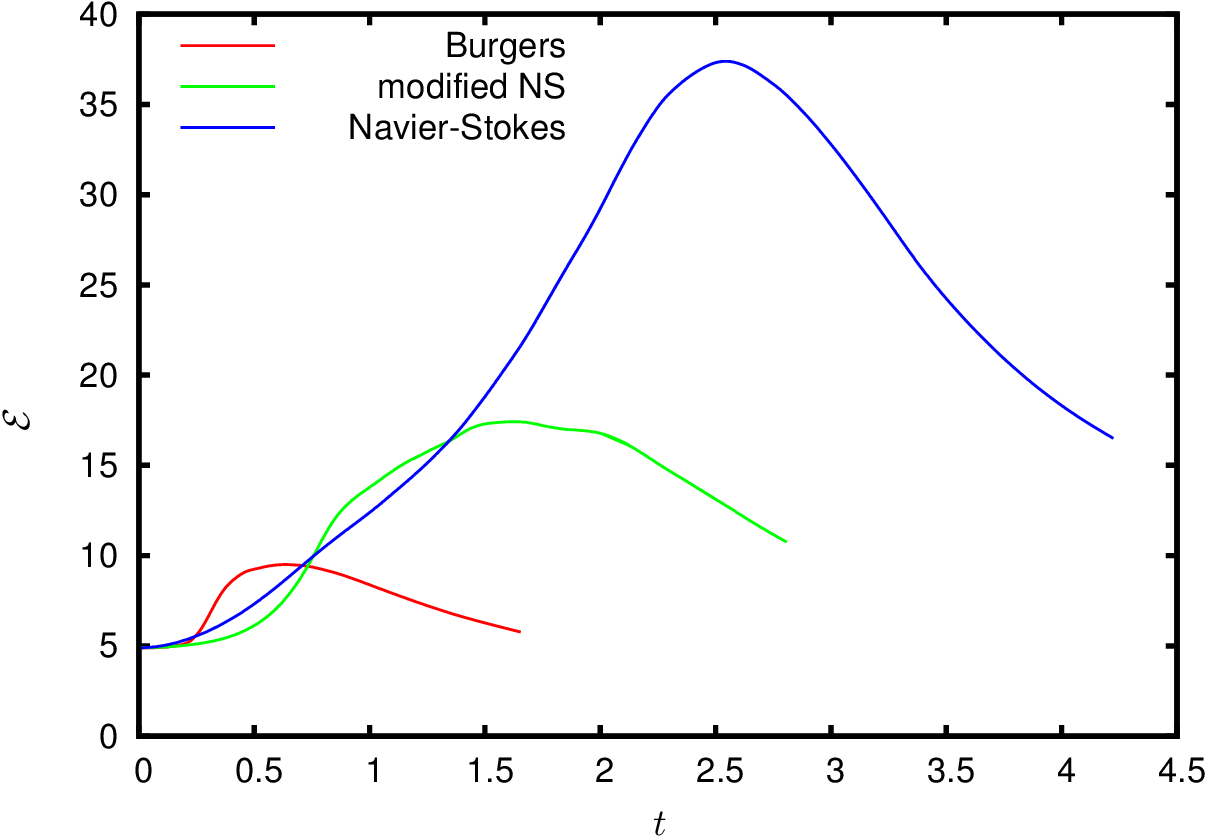} 
  \caption[Comparison of energy dissipation for Burgers equation and
  the proposed model and the time-development of enstrophy for Burgers
  equation]{\textbf{Left:} Comparison of the kinetic energy for
    Navier-Stokes, Burgers and the proposed model equation. The steep
    discontinuities of a Burgers flow explain the fast energy
    dissipation. \textbf{Right:} Evolution of free-falling turbulence
    for Burgers equation, the proposed model equation and
    Navier-Stokes equation. The turbulent flow is fully developed at
    $t=t_\mathcal{E}$, when the enstrophy $\mathcal{E}$ reaches its
    maximum.}
  \label{fig:Tenergydecay}
\end{center}
\end{figure}

Figure~\ref{fig:Tenergydecay} shows the decay of kinetic energy for
the considered hydrodynamical models in comparison. The tendency of
Burgers turbulence to form shocks and the dissipative nature of these
structures lead to a faster energy decay compared to the Navier-Stokes
equation. The new model equation exhibits a less violent form of
dissipation, its energy decay lies in between the others. The
difference in turbulence development is identified in a more precise
way when comparing the time $t_\mathcal{E}$ of maximum enstrophy
$\mathcal{E} = \int_{\Omega} \gvec{\omega}^2 \mathrm{d}x$. As Figure
\ref{fig:Tenergydecay} (right) indicates, the enstrophy of Burgers
turbulence reaches its peak significantly faster than for the
Navier-Stokes equation, in which vortex filaments dominate the
turbulent flow. The proposed model equation ranges between them. This
hints at the development of coherent structures at a timescale slower
than shock-formation of Burgers equation but faster than the formation
of vortex tubes for Navier-Stokes equation. Precise values for the
timescales are stated in Table~\ref{tab:runs}.

\begin{figure}[h]
  \centering
  \includegraphics[width=.3\textwidth]{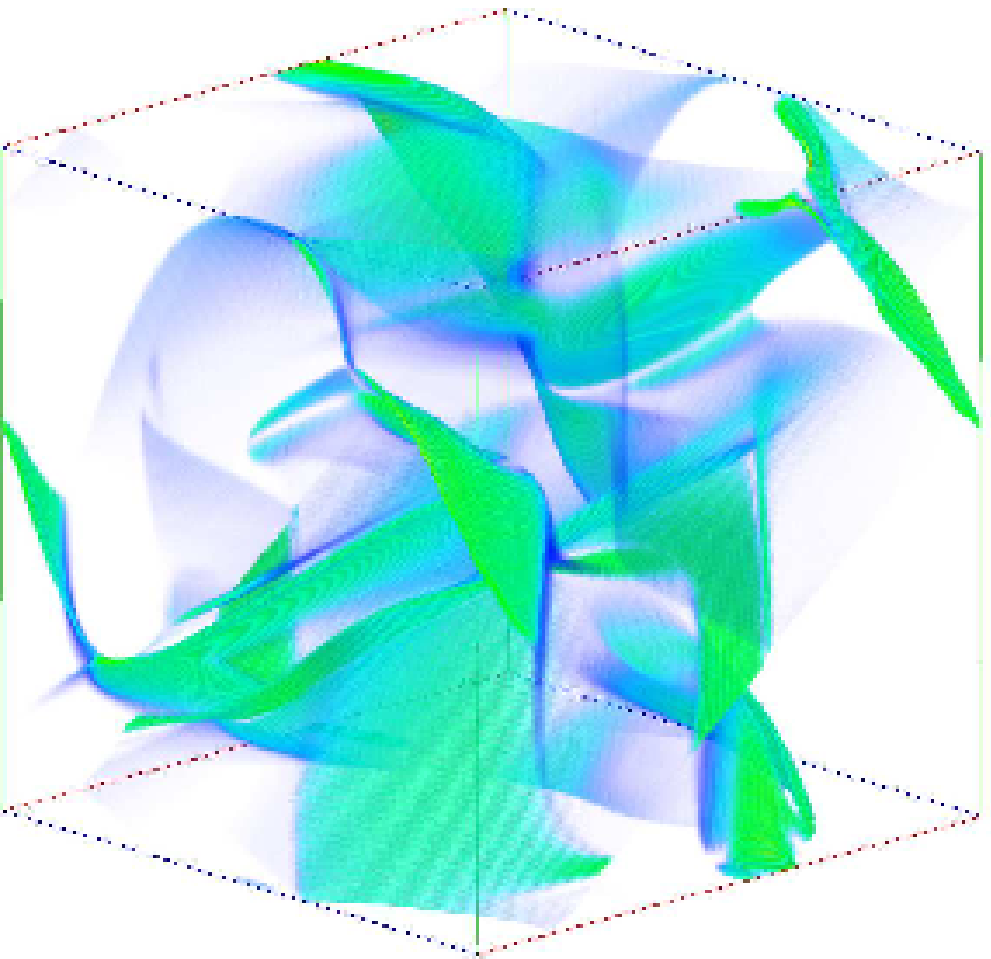}
  \includegraphics[width=.3\textwidth]{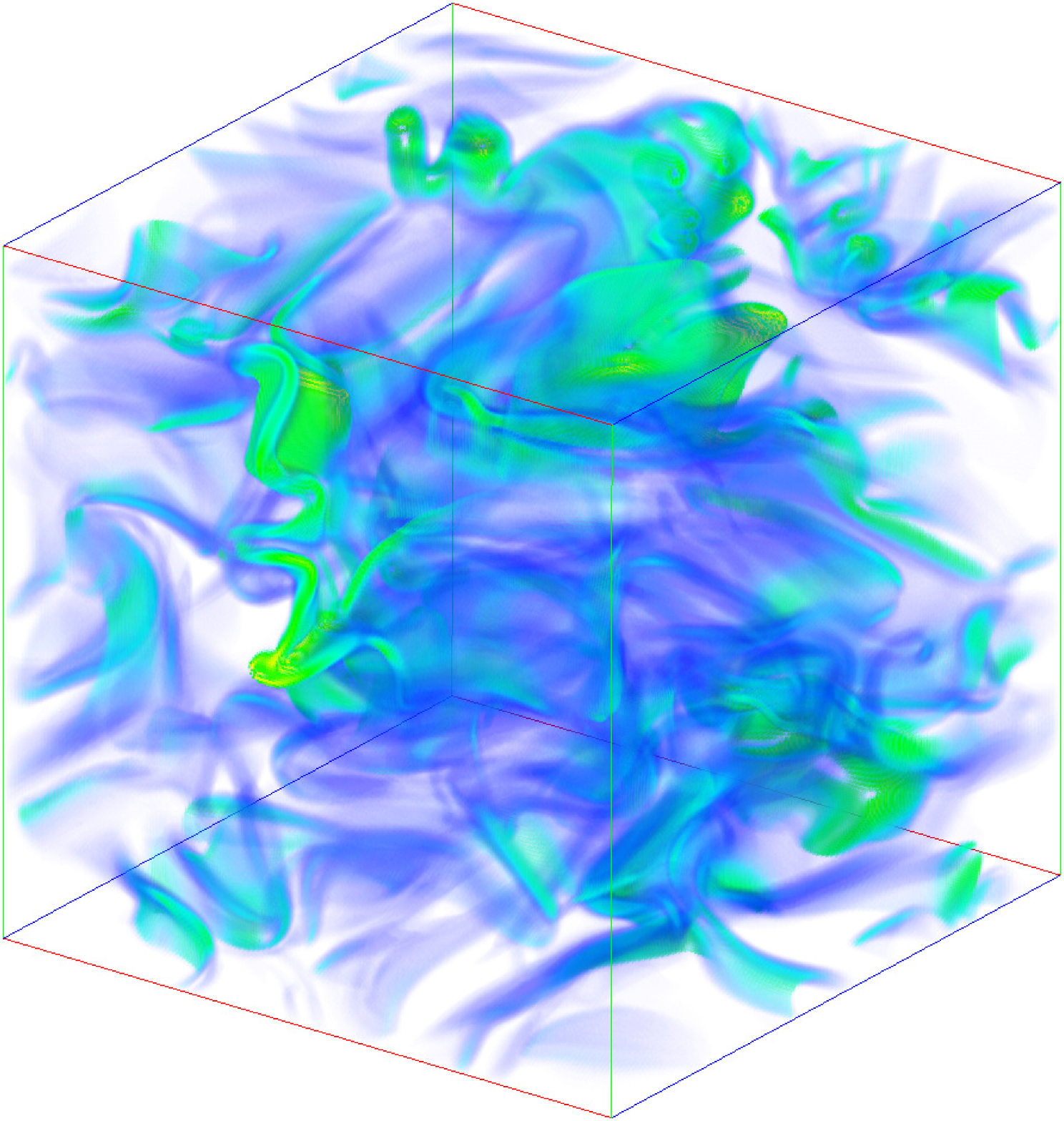}
  \includegraphics[width=.3\textwidth]{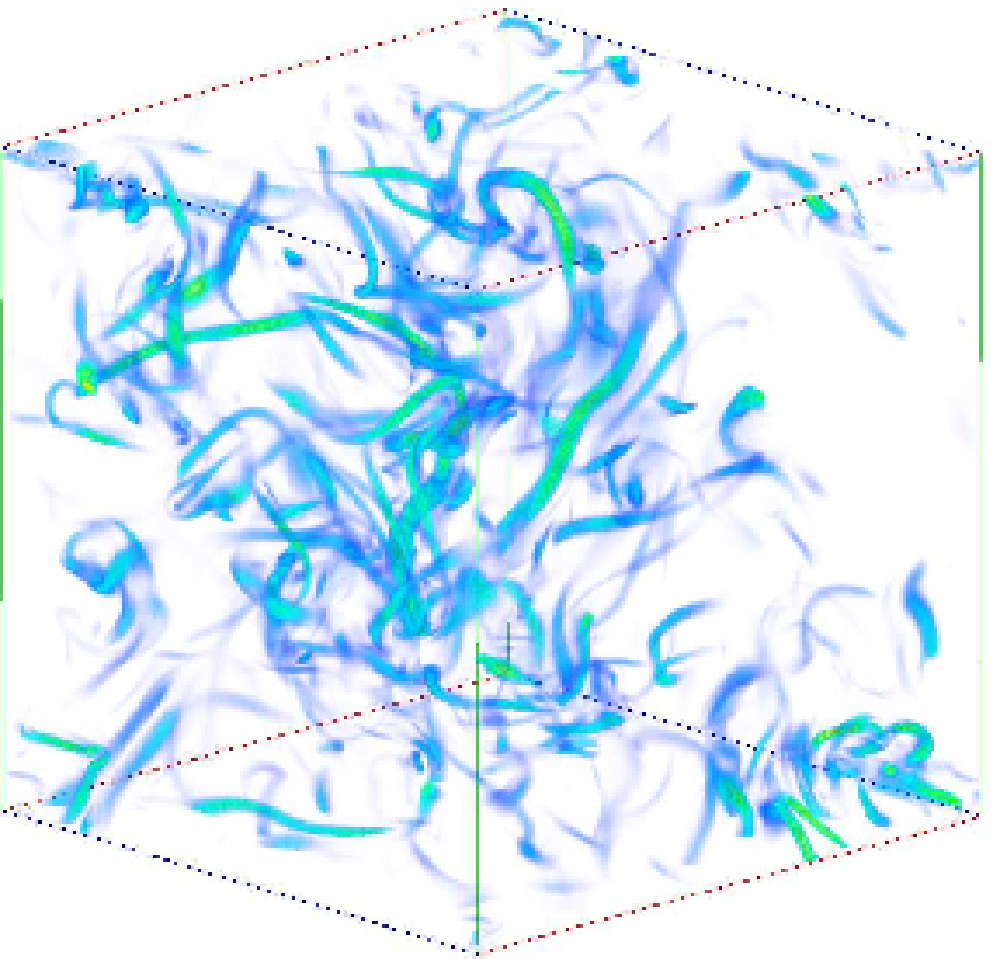}
  \caption{Dissipative structures in fully developed turbulence for
    different hydrodynamical models. \textbf{Left:} 2-dimensional
    shock fronts in Burgers turbulence. \textbf{Middle:} Folded vortex
    sheets in the new model equation. \textbf{Right:} Vortex filaments
    in Navier-Stokes turbulence.}
  \label{fig:vol}
\end{figure}

Figure~\ref{fig:vol} depicts a volume render of the fully developed
turbulence, the snapshot in each case taken at $t=t_\mathcal{E}$. As expected,
the Burgers flow (left) is dominated by shocks and the Navier-Stokes
flow (right) consists of vortex filaments. For the proposed model
equation (middle), the most dominant structures are two-dimensional
folded vortex-sheets.

A phenomenological description, which takes into account the most dissipative
structures of the turbulent flow, is the model of She and L\'ev\^eque
\cite{she-leveque:1994} connected to log-Poisson statistics of the local energy
dissipation \cite{dubrulle:1994}. They state that the scaling exponent
$\zeta_p$ of the p-th structure function behaves like
\begin{equation}
\label{eq:generalsheleveque}
\zeta_p = \frac{(1-k)p}{3} + C_0\left(1-\left(\frac{C_0 - k}{C_0}\right)^\frac{p}{3}\right),
\end{equation}
where $C_0$ is the co-dimension of the most dissipative structures in
the evolved flow and $k$ is the time-scaling exponent. This formula
will be referred to as the \textit{She-L\'ev\^eque model} even though in
\cite{she-leveque:1994} it is applied exclusively to the Navier-Stokes
equation (where $C_0 = 2,\; k=2/3$).

The typical time $t_l$ for the evolution of discontinuities for
turbulent Burgers flows scales linearly with $l$, which accounts for
$k=1$. As the shocks traveling through the domain are two-dimensional
we furthermore obtain $C_0=1$. Inserting this into equation
(\ref{eq:generalsheleveque}) leads to $\zeta_p = 1$. Since parts of the
velocity field are continuous, for $p<1$ the smoother regions of the
velocity field are pronounced. This would equal a scaling exponent of
$\zeta_p = p$ for these orders. Since this result is smaller than
$\zeta_p=1$, it is dominant for $p<1$. Therefore, we get
\begin{equation}
\zeta_p=
    \begin{cases}
      p & \text{for} p < 1\\
      1 & \text{for} p \ge 1
    \end{cases}
\label{eq:sheleveque_B}
\end{equation}
for Burgers equation.

\begin{table}[h]
  \centering
  \begin{tabular}{c|cccccc}
    p & 1 & 2 & 4 & 5 & 6 & 7 \\
    \hline
    $\alpha_p$ & 0.70 & 0.95 & 1.00 & 1.00 & 0.98 & 1.00 \\
    $\zeta_p$ & 1.00 & 1.00 & 1.00 & 1.00 & 1.00 & 1.00 \\    
  \end{tabular}
  \caption[Scaling exponents $\zeta_p$ for Burgers equation.]
  {Scaling exponents $\zeta_p$ for Burgers equation.
    The $\alpha_p$ measured with ESS are compared to the
    $\zeta_p$ predicted by the She-L\'ev\^eque
    model.}
  \label{tab:zetap_B}
\end{table}

Table~\ref{tab:zetap_B} shows the measurements of a direct numerical
simulation with $512^3$ grid points. Here, $\alpha_p$ is the data
obtained via ESS and $\zeta_p$ is the prediction of equation
(\ref{eq:sheleveque_B}). The visualization of this result is shown in
figure~\ref{fig:zetap} on the left. Especially for high order of $p$
($p>3$) the scaling exponent agrees with the prediction, yet smears
out for smaller $p$ (see also \cite{mitra-bec-pandit-frisch:2005}).

\begin{figure}[h]
  \centering
  \includegraphics[width=0.49\textwidth]{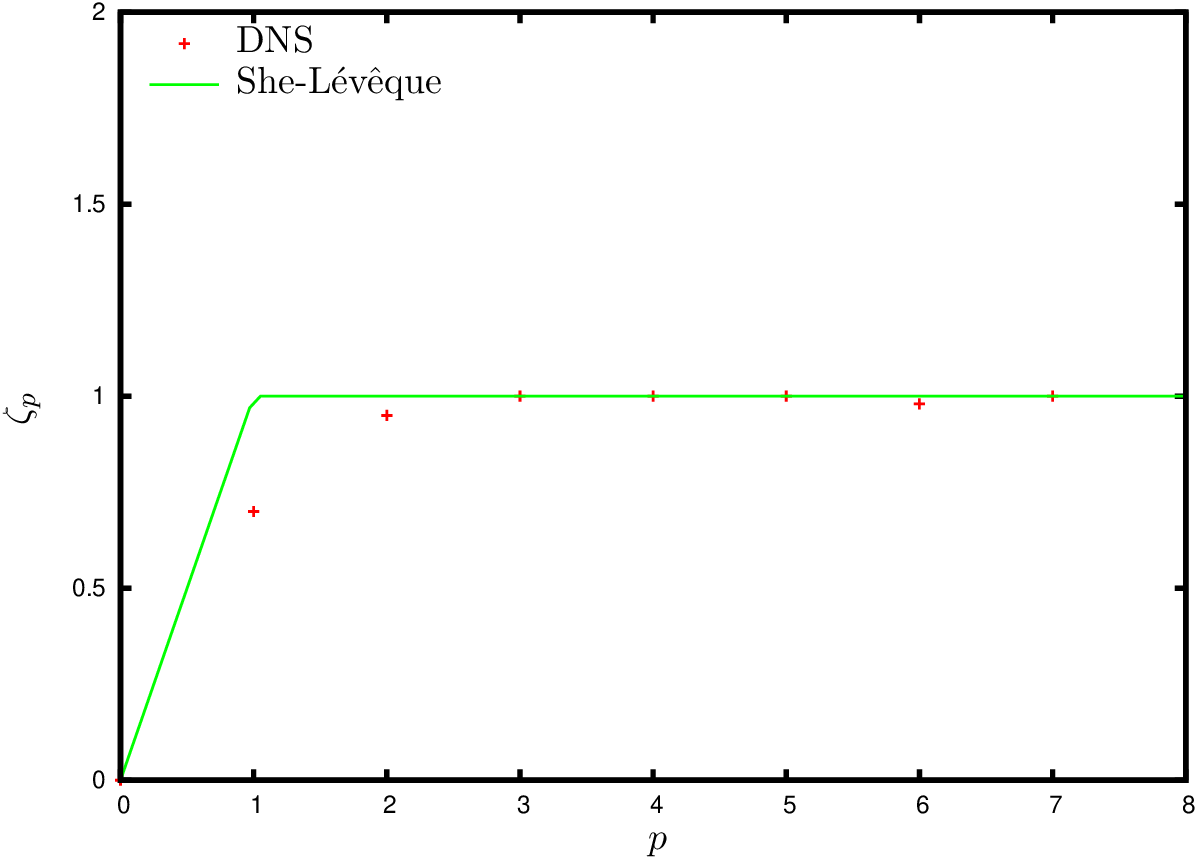}
  \includegraphics[width=0.49\textwidth]{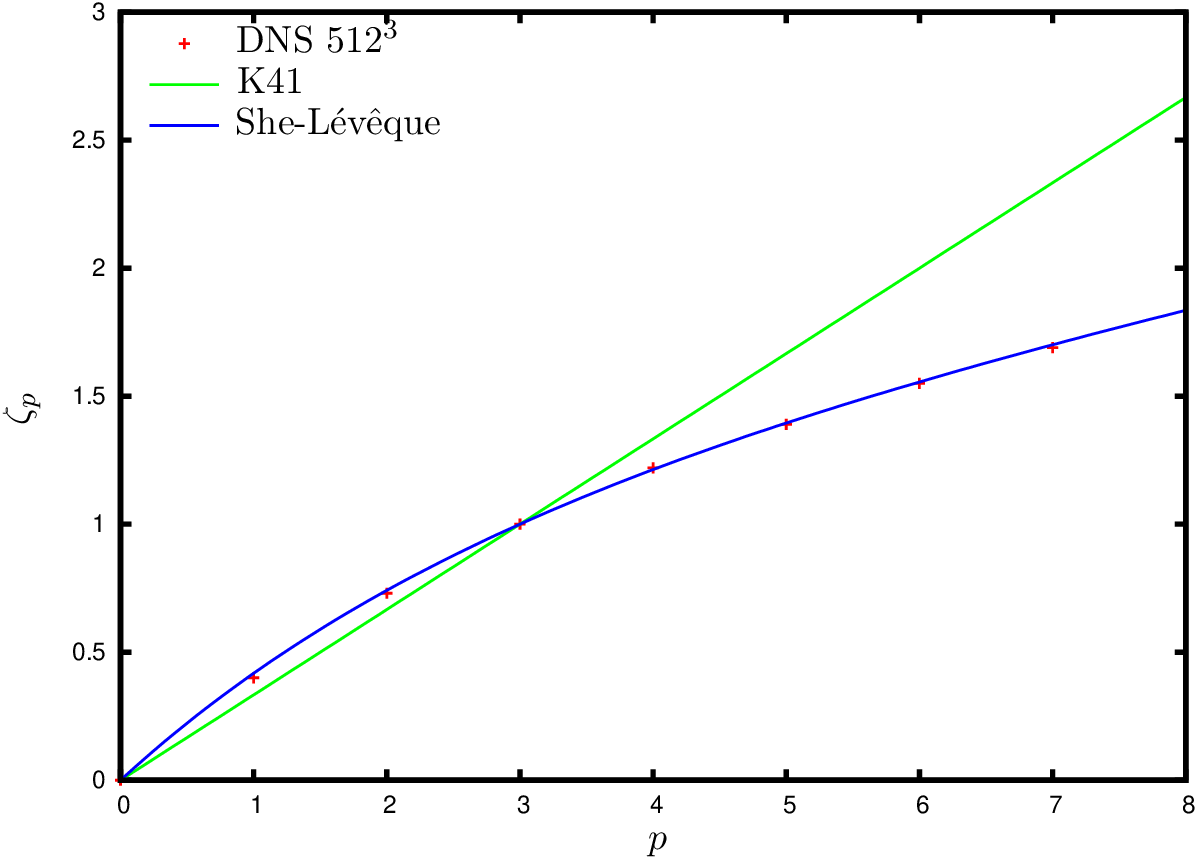}  
  \caption{Scaling exponents for fully developed turbulence,
    comparison between theory and numerical simulation. \textbf{Left:}
    Burgers equation. \textbf{Right:} Presented model equation.}
  \label{fig:zetap}
\end{figure}

Similar to Burgers equation the introduced model equation allows for
shocks to develop, since no incompressibility condition is
stated. Nevertheless, energy decays solely via the dissipative
term. Because of this, the hypotheses of Kolmogorov may be adapted to
Euler-Burgers equation. The Richardson cascade as well as the
properties of the energy spectrum and the scaling behavior of
structure functions should agree with the conclusions of Kolmogorov
and the She-L\'ev\^eque model. On the other hand, the structures that
evolve seem to be significantly different from the vortex filaments
known from Navier-Stokes.

Figure~\ref{fig:vol} (middle) shows the most dissipative structures of
fully developed turbulent flow. The norm of the vorticity visualizes
two-dimensional folded vortex sheets as the structures that correspond
to the vortex tubes of Navier-Stokes. This suggest $C_0=1$. The
time-scaling exponent $t_l$ for the introduced equation is estimated
as $k = 2/3$, with the same reasoning as for Navier-Stokes. Thus,
\begin{equation}
\zeta_p = \frac{p}{9} + \left(1-\left(\frac{1}{3}\right)^\frac{p}{3}\right)
\end{equation}
is the prediction for the scaling exponents proposed by the
She-L\'ev\^eque model.

\begin{table}[h]
  \centering
  \begin{tabular}{c|cccccc}
    p & 1 & 2 & 4 & 5 & 6 & 7 \\
    \hline
    $\alpha_p (\mathrm{P}\vec{u})$ & 0.4 & 0.73 & 1.22 & 1.39 & 1.55 & 1.69 \\
    $\zeta_p$ & 0.418 & 0.751 & 1.213 & 1.395 & 1.556 & 1.701 \\
  \end{tabular}
  \caption[Scaling exponents $\zeta_p$ for new model equation]{Scaling exponents 
    $\zeta_p$ for the new model equation. The $\alpha_p$ for $\mathrm{P} 
    \vec{u}$ measured with ESS are compared to the $\zeta_p$ predicted by 
    the She-L\'ev\^que model.}
  \label{tab:zetapEB}
\end{table}

Table~\ref{tab:zetapEB} features the results measured via ESS for the
scaling exponents of $\mathrm{P} \vec{u}$. As can be seen in figure
\ref{fig:zetap} (right) the numerical data agrees very well with the
prediction of the She-L\'ev\^eque model. The solenoidal
field is consistent with theory from low orders of $p$ up to the
highest order that was measured.


\section{Global solutions}
\label{sec:global}

In this section, we show global regularity for equation
(\ref{eq:eb}) for suitable initial data without any size restrictions. For this,
we proof the remarkable property, that this equation possesses an infinite
number of conserved quantities in the inviscid limit. Especially, the finiteness
of the $L^4$-norm of the velocity field coupled with
standard estimates for the $H^1$-norm enables us to show global regularity.

The problem of whether the three-dimensional incompressible
Navier-Stokes equations can develop a finite time singularity from
smooth initial conditions or if it has global solutions remains
unresolved (see
\cite{temam:1984,constantin-foias:1988,majda-bertozzi:2002} and the
references therein). The answer to this important question is
recognized as one of the Millennium prize problems
\cite{fefferman:2006,ladyzhenskaya:2003}.

Despite the complexity of the topic, a lot of progress has been made
on this field in the past. For the two-dimensional case, global-in-time
existence of unique weak and strong solutions is well-known (see
\cite{temam:1984,constantin-foias:1988}). In three dimensions weak
solutions are known to exist globally in time. For strong solutions,
existence and uniqueness is known for a short interval of time which
depends continuously on the initial data \cite{kato:1970}. Many
results published in the past, starting with \cite{serrin:1962},
provide criteria for the global regularity of solutions via conditions
applied to the velocity field \cite{berselli:2002,cao-titi:2008} or
components thereof \cite{kukavica-ziane:2005}, the vorticity
\cite{beale-kato-majda:1984}, its direction
\cite{constantin-fefferman:1993} or to the pressure field
\cite{chae-lee:2001,berselli-galdi:2002}.

The theory for the compressible Navier-Stokes equation is less well
developed, and we will not attempt a summary here. The
multi-dimensional Burgers equation \cite{burgers:1974} can be regarded
as a crude simplification of this model.  Global existence and
uniqueness of strong solutions can be established in two and
three-dimensions for suitably small initial conditions, much as with
the Navier-Stokes system. Irrotational flows do possess global
solutions for large data in arbitrary dimension, thanks to the
Cole-Hopf transformation \cite{cole:1951,hopf:1950}. However, there is
no multidimensional weak theory because of the absence of a mechanism
for energy dissipation, unlike Navier-Stokes.

The situation for this new modified Navier-~Stokes like equation is rather
different. In this section, the existence of global solutions is proven for the
model equation (\ref{eq:eb}) in a domain $\Omega$ which shall either be $\rr^3$
or a periodic cube in $\rr^3$.

\begin{theorem}
  \label{th:main}
  Let $\vec{u}_0\in H^1(\Omega)$.  Let $\vec{f}\in
  L^2_{loc}(\rr^+,L^2(\Omega))\cap L^1_{loc}(\rr^+,L^4(\Omega))$.
  Then the initial value problem for the model equation (\ref{eq:eb})
  has a unique global solution
  \[
  u\in C(\rr^+,H^1(\Omega))\cap L^2_{loc}(\rr^+,H^2(\Omega))\;.
  \]
\end{theorem}
The aim is to show that the solution remains {\em a priori} bounded in
$L^\infty([0,T),H^1(\Omega))\cap L^2([0,T),H^2(\Omega))$ for any
$T>0$, which implies its existence and uniqueness with standard
arguments comparable to e.g.  \cite{temam:1984,temam:1988}.
Throughout the argument, we denote the Euclidean norm of the vector
$\vec{u} = \sum_i u_i \vec{e}_i$ by $u= (\sum_i u_i^2)^{1/2}$.
We first prove the following lemma:
\begin{lemma}
  \label{lem:n}
  Let $\vec{u}_0$, $\vec{f}$, $\Omega$ be defined as above. Then the
  quantity $\|\uu(t)\|_{L^{p}}$ remains finite for $2 \le p \in \rr$.
\end{lemma}
\begin{proof}
  Taking the Euclidean inner product of (\ref{eq:eb}) with $\uu$
  yields the identity
  \begin{equation}
    \label{i1}
    \frac12\left(\frac{\partial}{\partial t}u^2+\pp\uu\cdot\nabla u^2\right)=\frac\nu2\Delta u^2 
    -\nu |\nabla \uu|^2 + \ff\cdot\uu\;.
  \end{equation}
  Integrate (\ref{i1}) over $\Omega$, use the fact that $\pp\uu$ is divergence free, and then apply  
  the Cauchy-Schwarz inequality to obtain
  \begin{equation}
    \frac12\frac{\partial}{\partial t} \|\uu\|\ellts+\nu\|\nabla \uu\|\ellts\le\|\ff\|\ellt\| \uu\|\ellt\;.
  \end{equation}
  Defining $x(t)=\frac12\left(\|\uu(t)\|\ellts+\int_0^t\nu\|\nabla \uu(s)\|\ellts ds\right)$, we have that
  \[
  x'(t)\le \|\ff(t)\|\ellt(2x(t))^{1/2}\;.
  \]
  Upon integration, this gives the inequality
  \begin{equation}
    \label{elltwo}
    \|\uu(t)\|\ellts+\int_0^t\nu\|\nabla \uu(s)\|\ellts ds \le\left( \|\uu_0\|\ellt+\int_0^t\|\ff(s)\|\ellt ds \right)^2\;.
  \end{equation}
  With this estimate the lemma is shown for the case $p=2$.

  Let $2 \le n \in \rr$ and multiply the identity (\ref{i1}) by $u^{2(n-1)}$:
  \begin{eqnarray*}
    \frac1{2n}\left(\frac{\partial}{\partial t}u^{2n}+\pp\uu\cdot\nabla u^{2n}\right)
    &=\frac\nu2\left(\frac{\partial}{\partial x_j}(u^{2(n-1)}\frac{\partial}{\partial x_j} u^2 )-\frac{4(n-1)}{n^2}|\nabla u^n|^2\right)\\
    &\quad-\nu u^{2(n-1)} |\nabla \uu|^2 + u^{2(n-1)}\ff\cdot\uu\;.
  \end{eqnarray*}
  Integrate this over $\Omega$ and apply H\"older's inequality:
  \[
  \frac1{2n}\frac{\partial}{\partial t}\|\uu\|_{L^{2n}}^{2n}+\int_\Omega\left(\textstyle{\frac{2\nu(n-1)}{n^2}} \left|\nabla u^n\right|^2
    +\nu u^{2(n-1)} |\nabla \uu|^2\right)d\mathbf{x}\le \|\ff\|_{L^{2n}}\|\uu\|_{L^{2n}}^{2n-1}
  \]
  If we let
  \[
  y(t)=\frac1{2n}\|\uu(t)\|_{L^{2n}}^{2n}+\int_0^t\int_\Omega\left(\textstyle{\frac{2\nu(n-1)}{n^2}} \left|\nabla u^n(s)\right|^2
    +\nu u^{2(n-1)} (s)|\nabla \uu(s)|^2\right)d\mathbf{x}ds\;,
  \]
  then we obtain
  \[
  y'(t)\le \|\ff(t)\|_{L^{2n}} (2n\;y(t))^{\frac{2n-1}{2n}}\;.
  \]
  This leads to the estimate
  \begin{eqnarray}
    \label{elltn}
    \|\uu(t)\|_{L^{2n}}^{2n}+2n\int_0^t\int_\Omega&\left(\textstyle{\frac{2\nu(n-1)}{n^2}} \left|\nabla u^n(s)\right|^2
      +\nu u^{2(n-1)} (s)|\nabla \uu(s)|^2\right)d\mathbf{x}ds\\
    \nonumber
    &\le\left(\|\uu_0\|_{L^{2n}}+\int_0^t\|\ff(s)\|_{L^{2n}}ds\right)^{2n}\;.
  \end{eqnarray}
\end{proof}
\noindent
\textbf{Remark:} This key argument fails for the case of the
Navier-Stokes equation. At the same time, this estimate 
establishes an infinite number of conserved quantity in the
unforced inviscid case.\\

\begin{proof}[Proof of Theorem \ref{th:main}]
  Take the $L^2$-inner product of (\ref{eq:eb}) with $\Delta
  \uu$ and integrate by parts:
  \[
  \frac12\frac{\partial}{\partial t} \|\nabla \uu\|\ellts+\nu\|\Delta \uu\|\ellts
  =\underbrace{\int_\Omega (\pp\uu\cdot\nabla \uu)\cdot\Delta\uu \;d\mathbf x}_{\mathrm (i)}
  +\underbrace{\int_\Omega\ff\cdot\Delta \uu \;d\mathbf x}_{\mathrm(ii)}\;.
  \]
  The forcing term (ii) has the bound
  \[
  \int_\Omega\ff\cdot\Delta \uu d\mathbf x\le \|\ff\|\ellt\|\Delta \uu\|\ellt\le \frac\nu4\|\Delta \uu\|\ellts+\frac 1\nu\|\ff\|\ellts\;.
  \]
  The nonlinear term (i) is estimated as follows:
  \begin{eqnarray*}
    \int_\Omega (\pp \uu\cdot\nabla \uu)\cdot\Delta \uu\; d\mathbf x
    &=-\int\frac{\partial}{\partial x_k}u_i\frac{\partial}{\partial x_k}\left((\pp \uu)_j\frac{\partial}{\partial x_j}u_i\right)\; d\mathbf x\\
    &=-\int\frac{\partial}{\partial x_k}u_i\left((\pp \uu)_j\frac{\partial}{\partial x_j}\frac{\partial}{\partial x_k}u_i
      +\frac{\partial}{\partial x_k}(\pp \uu)_j\frac{\partial}{\partial x_j}u_i\right)\; d\mathbf x\\
    &=-\int\left(\frac12 (\pp \uu)_j\frac{\partial}{\partial x_j}|\nabla \uu|^2 
      + \frac{\partial}{\partial x_k}u_i\frac{\partial}{\partial x_j}(\frac{\partial}{\partial x_k}(\pp \uu)_ju_i)\right)\; d\mathbf x\\
    &=\int\frac{\partial}{\partial x_j}\frac{\partial}{\partial x_k}u_i\frac{\partial}{\partial x_k}(\pp \uu)_ju_i\; d\mathbf x\\
    &\le\|\nabla^2\uu\|\ellt\|\nabla \pp \uu\|_{L^4}\|\uu\|_{L^4}\;.
  \end{eqnarray*}
  The second norm above is handled by interpolation.  We first note
  that
  \[
  \|\nabla \pp \uu\|_{L^4}\le \|\nabla \pp \uu\|_{L^6}^{3/4}\|\nabla \pp \uu\|_{L^2}^{1/4}\;.
  \]
  Now when $\Omega=\rr^3$, the Sobolev embedding theorem gives 
  \begin{equation}
    \label{ellsix}
    \|\nabla \pp \uu\|_{L^6}\le C\|\nabla^2 \pp \uu\|\ellt\;.
  \end{equation}
  When $\Omega$ is a periodic domain, the norm on the right must be
  replaced by $\|\nabla \pp \uu\|_{H^1}$.  However, since $\nabla \pp
  \uu$ has zero mean, this is bounded again by $C\|\nabla^2 \pp
  \uu\|\ellt$, by the Poincar\'e inequality. Therefore, (\ref{ellsix})
  holds in both cases. Using the facts that the operator $\pp$
  commutes with derivatives and that it is a projection in $L^2$, we
  have that
  \[
  \|\nabla \pp\uu\|\ellt\le \|\nabla\uu\|\ellt 
  \quad\mbox{and}\quad
  \|\nabla^2 \pp\uu\|\ellt\le \|\nabla^2\uu\|\ellt\;.
  \]
  Next, we use integration by parts to obtain the simple ellipticity
  identity
  \begin{eqnarray}
    \label{elliptic}
    \|\nabla^2 \uu\|\ellts
    &=\int_\Omega \frac{\partial}{\partial x_j}\frac{\partial}{\partial x_k}u_i
    \frac{\partial}{\partial x_j}\frac{\partial}{\partial x_k}u_i \;d\mathbf{x}\\
    \nonumber
    &=\int_\Omega \frac{\partial}{\partial x_j}\frac{\partial}{\partial x_j}u_i
    \frac{\partial}{\partial x_k}\frac{\partial}{\partial x_k}u_i \;d\mathbf{x}\\
    \nonumber
    &=\|\Delta \uu\|\ellts\;.
  \end{eqnarray}
  Combining these observations with Young's inequality, we conclude
  that the nonlinear term (i) is bounded by
  \[
  C\|\Delta \uu\|\ellt^{7/4}\|\nabla\uu\|\ellt^{1/4}\|\uu\|_{L^4}
  \le \frac{\nu}{4}\|\Delta\uu\|\ellt^2+\frac{C}{\nu^7}\|\nabla\uu\|\ellt^2\|\uu\|_{L^4}^8\;.
  \]
  Altogether, we get the inequality
  \[
  \frac{\partial}{\partial t}\|\nabla \uu\|\ellts+\nu\|\Delta \uu\|\ellts\le \frac {C}{\nu^7}\|\nabla \uu\|\ellt^2\|\uu\|_{L^4}^8+\frac C\nu\|\ff\|\ellts\;.
  \]
  Using Gronwall's inequality, we find that
  \begin{eqnarray}
    \label{b3}
    \|\nabla \uu\|\ellts+\nu\int_0^t\|\Delta \uu(s)\|\ellts \;ds \le& \|\nabla \uu_0\|\ellts\exp\frac {C}{\nu^7}\int_0^t\|\uu(s)\|_{L^4}^8 \; ds\\
    \nonumber
    &+\frac C\nu\int_0^t \left(\exp\frac C{\nu^7}\int_s^t\|\uu(\sigma)\|_{L^4}^8\;d\sigma\right)\|\ff(s)\|\ellts\;ds\;.
  \end{eqnarray}
  Combining (\ref{elltwo}) and Lemma \ref{lem:n} with $p=4$, and
  (\ref{b3}), we see that the quantity
  \[
  \|\uu(t)\|^2_{H^1}+\int_0^t\nu[\|\nabla\uu(s)\|^2\ellt+\|\Delta\uu(s)\|^2\ellt]\;ds
  \]
  remains finite.  However, by (\ref{elliptic}) and the fact that 
  $L^\infty_{loc}(\rr^+,L^2(\Omega)) \subset L^2_{loc}(\rr^+,L^2(\Omega))$
  we have that
  \[
  \|\uu(t)\|^2_{H^1}+\int_0^t\nu\|\uu(s)\||^2_{H^2(\Omega)}\;ds
  \]
  also remains finite.
\end{proof}


\section{Final remarks}

In this paper, a modified Navier-Stokes equation is presented. Its
dynamics and turbulent behavior are studied in terms of the scaling
properties of its structure functions. The most dissipative structures
are identified as vortex sheets of co-dimension~1, which allows us to
compare the numerically measured scaling exponents to a modified
phenomenologically based She-L\'ev\^eque approach.

Furthermore, we prove the existence of global solutions for this
equation. A remarkable consequence of this proof is the existence of an
infinite number of conserved quantities $\|\mathbf{u}\|_{L^p}$ in the
ideal (non-viscous) case without forcing. This property is not only
responsible for the existence of global solutions but should show up
in the statistics of intermittent turbulent fluctuations. Work in this
direction is in progress.

\section*{Acknowledgment}

This work benefited from support through DFG-GR 967/3-1.


\bibliographystyle{model1-num-names}

\end{document}